\documentclass[twocolumn,aps,prx,10pt,nofootinbib]{revtex4-2}

\usepackage[dvipsnames]{xcolor}
\usepackage[colorlinks=true,linkcolor=Blue,urlcolor=Blue,citecolor=Blue]{hyperref}

\usepackage{amsfonts,amsmath,amssymb,amsthm}
\usepackage{physics,mathtools,bbm}

\newtheorem{Theorem}{Theorem}
\newtheorem*{Theorem*}{Theorem}
\newtheorem{Corollary}{Corollary}

\begin{document}

\title{\texorpdfstring{Quantum codes are quantized codes:\\Stabilizer codes from quantization over finite fields}{Quantum codes are quantized codes: Stabilizer codes from quantization over finite fields}}
\author{Michael Zurel}
\affiliation{Institute for Quantum Computing, University of Waterloo, Waterloo, Ontario N2L 3G1, Canada}
\affiliation{Department of Applied Mathematics, University of Waterloo, Waterloo, Ontario N2L 3G1, Canada}
\date{\today}

\begin{abstract}
    We present a discrete analogue of geometric quantization over finite fields which, when applied to certain classical error-correcting codes, yields quantum error-correcting codes. This provides a physical interpretation for a previously known quantum-classical correspondence in coding theory.
\end{abstract}

\maketitle

The objective of a quantization procedure is to take a theory describing a classical physical system and produce a corresponding quantum system. Often, the classical system in question is a Hamiltonian system over a generalized position-momentum phase space. A quantization procedure then provides a mapping from the algebra of classical observables (functions on the phase space) to a corresponding algebra of quantum observables (operators on a Hilbert space).

Canonical quantization proceeds by mapping canonical position and momentum variables, $q_i$ and $p_j$, to operators, $\hat{q}_i$ and $\hat{p}_j$, satisfying the canonical commutation relations, $[\hat{q}_i,\hat{p}_j]=i\hbar\delta_{ij}$. Functions on phase space are then mapped to operators on Hilbert space by replacing the variables $q_i$ and $p_j$ by the corresponding operators. But since $\hat{q}_i$ and $\hat{p}_j$ generally do not commute, there is an ambiguity in the operator that gets assigned to a function coming from the ordering of these operators. Different choices, such as standard ordering ($\hat{q}$'s left of $\hat{p}$'s), normal ordering (creation operators left of annihilation operators), or symmetric orderings like in Weyl quantization~\cite{Weyl1927}, lead to different quantization procedures.

Geometric quantization~\cite{Souriau1966,Kostant1970} reframes quantization in geometric terms. Instead of starting with a preferred set of canonical coordinates, we start with the classical phase space itself, viewed as a symplectic manifold. The construction then proceeds in steps known as \emph{prequantization}, \emph{polarization}, and \emph{quantization}. This does not eliminate all choices, for example, there are different ways of choosing a polarization, but it organizes the ambiguities in a way that is tied to the geometry of the phase space, rather than to ad hoc choices of operator orderings.

One can ask what properties a reasonable quantization map should have. Some natural requirements are: (1)~real-valued functions map to self-adjoint operators, (2)~constant functions map to corresponding multiples of the identity operator, and (3)~the Hilbert space is not ``too big'' (for example, we could impose that the action of position and momentum operators be irreducible). Dirac proposed a further requirement~\cite{Dirac1958}, namely, that the Poisson bracket of the classical system be reproduced by the commutator of the quantum system via a correspondence of the form
\begin{equation}\label{Equation:PoissonBracketCommutator}
    i\hbar\mathcal{Q}(\{f,g\})=[\mathcal{Q}(f),\mathcal{Q}(g)],
\end{equation}
where $\mathcal{Q}(f)$ is the operator assigned to the function $f$. This rule has a number of motivations. First, taking the classical limit $\hbar\to0$ of the quantum system should reproduce the structure of the classical system, where multiplication of observables commutes. Second, the commutator is the structure on the Hilbert space that reproduces the algebraic properties of the Poisson bracket~\cite{Dirac1958}. This manifests in the fact that, where the Poisson bracket generates dynamics and symmetries in classical systems (e.g., via Hamilton's equations $\dot{f}=\{f,H\}$), the commutator plays the corresponding role in quantum systems (e.g., via Heisenberg's equation $\dot{F}=[F,H]/(i\hbar)$).

One might hope that a unique quantization map could be singled out by imposing enough well-motivated constraints, thus resolving the above mentioned ambiguity. However, this is not possible, since the Groenewold-van Hove theorem shows that there is typically no quantization map consistent with the above constraints~\cite{Groenewold1946,vanHove1951,Gotay1999}. Therefore, when constructing a quantization map, we must choose a subset of consistent constraints to satisfy.

Another common requirement for a quantization procedure is the capacity to handle constrained Hamiltonian systems. Many fundamental physical theories, including electromagnetism, Yang-Mills theory, and general relativity, are of this type. In Dirac's classification~\cite{Dirac1964}, so-called ``first-class'' constraints generate gauge symmetries, with physical degrees of freedom living in the symplectic reduction of the phase space~\cite{MarsdenWeinstein1974}, i.e., the space of gauge orbits in the constraint surface. A desirable property of a quantization procedure is that it be compatible with this structure, in the sense that quantizing the reduced phase space is equivalent to quantizing first and then imposing the corresponding quantum constraints. Geometric quantization often has this property~\cite{Gotay1986,AshtekarStillerman1986,GuilleminSternberg1982,TianZhang1998}.

Coding theory, both classical and quantum, exhibits a similar structure. A classical linear code is obtained by restricting to a subspace of a vector space over a finite field. This can be done by imposing a number of linear constraints (called parity checks)~\cite{MacWilliamsSloane1977}. Similarly, the code space of a quantum stabilizer code is obtained by restricting to the joint eigenspace of a set of commuting Pauli operators by imposing a set of stabilizer constraints~\cite{Gottesman1997}.

In this paper, we introduce a partial analogue of geometric quantization for symplectic vector spaces over finite fields. The construction has a similar structure to geometric quantization of finite-dimensional symplectic vector spaces over the reals, but with certain analytic concepts replaced by suitable finite-field theoretic analogues. When applied to a system with a set of linear first-class constraints, the result is the code space of a stabilizer code. The stabilizers and the logical code space are obtained by quantizing the constraints and the reduced phase space respectively. Additionally, we note that the classical constraints can be identified with the parity checks of a classical code, hence the tagline: \emph{quantum codes are quantized codes.} Finally, we find that a discrete analogue of the Segal-Bargmann representation~\cite{Bargmann1961,Segal1963} reproduces a well-known correspondence between stabilizer codes on $d$-dimensional qudits and certain classical codes over $\mathbb{F}_{d^2}$~\cite{CalderbankSloane1998,Rains1999,KetkarSarvepalli2006}.

\section*{Geometric quantization of symplectic vector spaces}

Geometric quantization is a general framework for mapping classical systems to quantum systems~\cite{Souriau1966,Kostant1970}. It starts with a symplectic manifold $(M,\omega)$, and proceeds in steps. First, \emph{prequantization} constructs a faithful representation of the Poisson algebra of classical observables as operators on a ``prequantum'' Hilbert space satisfying Eq.~\eqref{Equation:PoissonBracketCommutator}. Formally, \emph{prequantization} constructs a complex line bundle over $M$ whose curvature encodes the symplectic form, and takes ``wavefunctions'' to be square-integrable sections of this bundle. A classical observable $f$ is mapped to a differential operator constructed from the Hamiltonian vector field of $f$.

The prequantum Hilbert space is too large to represent the physical quantum system. The second step, \emph{polarization}, is a mechanism for selecting ``half the variables'' to define the configuration space, analogous to deciding that wavefunctions should depend on $q_j$ but not on $p_j$, or vice versa. Restricting to polarized sections yields a smaller Hilbert space and a corresponding class of observables. The resulting \emph{quantization} map takes classical observables to operators on the polarized Hilbert space.

One of the strengths of geometric quantization is that it can be applied to a very broad class of systems represented by symplectic manifolds, not just the typical cases like cotangent bundles of real manifolds. That said, for our purposes, we need only the basics of geometric quantization for the special case of symplectic vector spaces over $\mathbb{R}$. For a more thorough and pedagogical overview of geometric quantization see Refs.~\cite{Blau1992,Hall2013}. The notation used in this section is consistent with that of Ref.~\cite{Hall2013}.

\medskip

Let $(V,\omega)$ be a $2n$-dimensional symplectic vector space over $\mathbb{R}$, with Darboux coordinates $(q_1,\dots,q_n,p_1,\dots,p_n)$ and canonical $2$-form $\omega=\sum_jdp_j\wedge dq_j$. From the symplectic structure, we get the Poisson bracket,
\begin{equation}
    \{f,g\}=-\omega^{-1}(df,dg)=\sum_j\frac{\partial f}{\partial q_j}\frac{\partial g}{\partial p_j}-\frac{\partial f}{\partial p_j}\frac{\partial g}{\partial q_j}.
\end{equation}
In this case, $\omega$ is globally exact and the prequantum line bundle is trivial, so prequantization can be described directly in terms of complex-valued functions on $V$~\cite{Hall2013}.

\paragraph*{Prequantization.}

The prequantization map takes classical observables in $C^\infty(V)$ to operators on the prequantum Hilbert space $L^2(\mathbb{R}^{2n})$. To define this map, first choose a symplectic potential $\theta$, i.e., a $1$-form satisfying $d\theta=\omega$. Different choices of symplectic potential yield prequantization maps that are unitarily equivalent~\cite[Proposition~22.5]{Hall2013}. For concreteness, we can choose
\begin{equation}\label{Equation:SymplecticPotential}
    \theta=\sum\nolimits_jp_jdq_j.
\end{equation}
This defines a connection $\nabla=d-i\theta/\hbar$ on the trivial line bundle $V\times\mathbb{C}$ whose curvature is $-i\omega/\hbar$. For any $f\in C^\infty(V)$, let $X_f$ be its Hamiltonian vector field, defined by  $df=\omega(X_f,\cdot)$, which satisfies $X_f(g)=\{f,g\}=-\omega(X_f,X_g)$. Then the prequantization map is
\begin{equation}
    \mathcal{Q}_{\mathrm{pre}}(f)=i\hbar\nabla_{X_f}+f,\quad\forall f\in C^\infty(V),
\end{equation}
where $\nabla_{X_f}=X_f-i\theta(X_f)/\hbar$ is the covariant derivative in the direction of $X_f$.

This satisfies many of the properties desired for a quantization map: $\mathcal{Q}_{\mathrm{pre}}(1)=\mathbbm{1}$, Eq.~\eqref{Equation:PoissonBracketCommutator} holds for all $f,g\in C^\infty(V)$, and, subject to some qualifications, smooth real-valued functions are mapped to self-adjoint operators~\cite{Hall2013}. In particular,
\begin{equation}\label{Equation:PrequantumPositionMomentumOperators}
    \begin{cases}
        \mathcal{Q}_{\mathrm{pre}}(q_j)=i\hbar\frac{\partial}{\partial p_j}+q_j,\\
        \mathcal{Q}_{\mathrm{pre}}(p_j)=-i\hbar\frac{\partial}{\partial q_j}.
    \end{cases}
\end{equation}
The exponentiated versions of these operators are finite displacement operators, which act on wavefunctions as
\begin{equation}
    \begin{cases}
        \left(e^{it\mathcal{Q}_{\mathrm{pre}}(q_j)/\hbar}\Psi\right)(q,p)=e^{itq_j/\hbar}\Psi(q,p-te_j),\\
        \left(e^{it\mathcal{Q}_{\mathrm{pre}}(p_j)/\hbar}\Psi\right)(q,p)=\Psi(q+te_j,p).
    \end{cases}
\end{equation}

Any observable that is linear in the variables $q_i,p_j$ can be written uniquely as $f_v(\cdot):=\omega(v,\cdot)$ for some $v\in\mathbb{R}^{2n}$. These functions satisfy $X_{f_v}=v$, and $\{f_u,f_v\}=-\omega(u,v)$. Then, for each $v:=(q',p')\in\mathbb{R}^{2n}$, we can define a prequantum displacement operator associated to the linear observable $f_v$ as
\begin{equation}
    D_v=e^{i\mathcal{Q}_{\mathrm{pre}}(f_v)/\hbar}.
\end{equation}
Using the Zassenhaus formula, these operators act on prequantum wavefunctions as
\begin{equation}
    (D_{(q',p')}\Psi)(q,p)=e^{i(p'\cdot q-\frac{1}{2}q'\cdot p')/\hbar} \Psi(q-q',p-p')
\end{equation}
and they satisfy the Weyl relations
\begin{equation}
    \begin{cases}
        D_uD_v=e^{i\omega(u,v)/(2\hbar)}D_{u+v}\\
        D_{-v}=D_v^{-1}
    \end{cases}\quad\forall u,v\in V.
\end{equation}
That is, the operators $\{e^{it/\hbar}D_v\;|\;v\in V,t\in\mathbb{R}\}$ define a representation of the Heisenberg group $H(V)$, with composition law $(u,s)\circ(v,t)=(u+v,s+t+\omega(u,v)/2)$ for all $u,v\in V$, $s,t\in\mathbb{R}$.

However, the prequantum Hilbert space is, in a sense, too large to represent the physical quantum Hilbert space~\cite{Blau1992,Hall2013}. In particular, the representation of the Heisenberg group is reducible, a situation remedied by choosing a polarization.

\paragraph*{Polarization.}
There are different ways of choosing a polarization. The most common \emph{real} polarizations are the position and momentum polarizations. The position polarization is given by $\text{span}\{\partial/\partial p_j\}_j$, and the corresponding space of polarized wavefunctions consists of functions $\Psi\in C^\infty(V)$ satisfying $\nabla_{\partial/\partial p_j}\Psi=0$ for all $j$. With the choice of symplectic potential above, this means $\partial\Psi/\partial p_j=0$, so the functions $\Psi(q,p)=\psi(q)$ depend only on the position variables. Similarly, the momentum polarization is $\text{span}\{\partial/\partial q_j\}_j$, and the corresponding space of polarized wavefunctions consists of functions satisfying $\nabla_{\partial/\partial q_j}\Psi=0$. These are functions of the form $\Psi(q,p)=e^{iq\cdot p/\hbar}\phi(p)$ which, up to the exponential gauge factor, depend only on the momentum variables.

We get complex polarizations by identifying $\mathbb{R}^{2n}$ with $\mathbb{C}^n$ via
\begin{equation}\label{Equation:SegalBargmannCoordinates}
    z_j:=q_j-i\alpha p_j\quad\forall j
\end{equation}
for some $\alpha>0$. Then $\text{span}_{\mathbb{C}}\{\partial/\partial\bar{z}_j\}_j$ defines a polarization, where
\begin{equation}
    \frac{\partial}{\partial\bar{z}_j}=\frac{1}{2}\left(\frac{\partial}{\partial q_j}-\frac{i}{\alpha}\frac{\partial}{\partial p_j}\right)
\end{equation}
is the Wirtinger derivative with respect to the complex structure on $V$ defined in Eq.~\eqref{Equation:SegalBargmannCoordinates}. The corresponding space of polarized wavefunctions consists of functions satisfying $\nabla_{\partial/\partial\bar{z}_j}\Psi=0$ for all $j$. These functions have the form $\Psi(q,p)=F(z)e^{-\alpha|p|^2/(2\hbar)}$ where $F$ is an arbitrary holomorphic function on $\mathbb{C}^n$. This is the Segal-Bargmann representation~\cite{Bargmann1961,Segal1963}.

\paragraph*{Quantization.}

One issue with the above construction is that there are no nonzero square-integrable functions in $L^2(\mathbb{R}^{2n})$ satisfying the constraints of the position or momentum polarizations. In order to get the quantum Hilbert space, we need to equip the spaces of polarized wavefunctions with a suitably modified inner product. Formally, this is done by introducing so-called ``half-forms'' or ``metaplectic corrections'' (see Ref.~\cite{Hall2013} for details), but for our purposes, this is equivalent to simply equipping each subspace of polarized wavefunctions with the natural inner product on the corresponding space of independent variables.

For example, in the position polarization, the polarized wavefunctions of the form $\Psi(q,p)=\psi(q)$ can be identified with the Hilbert space $L^2(\mathbb{R}^n)$, with inner product given by integrating over only the position variables,
\begin{equation}
    \langle\Psi_1|\Psi_2\rangle=\int_{\mathbb{R}^n} dq\,\psi_1(q)^*\psi_2(q).
\end{equation}
The prequantum operators associated to the linear observables $q_j$ and $p_j$ preserve this subspace, and their induced action is the usual Schr\"odinger representation with $\hat{q}_j\psi=q_j\psi$ and $\hat{p}_j\psi=-i\hbar\partial\psi/\partial q_j$.

Similarly, the momentum polarization yields Hilbert space $L^2(\mathbb{R}^n)$, with inner product obtained by integrating over only the momentum variables, and with $\hat{p}_j$ acting by multiplication and $\hat{q}_j$ by $i\hbar\partial/\partial p_j$. The holomorphic polarization described above gives the Segal-Bargmann space of holomorphic functions on $\mathbb{C}^n$ equipped with its Gaussian-weighted inner product~\cite{Bargmann1961,Segal1963}.

Only the observables whose prequantum operators preserve the space of polarized wavefunctions can be restricted directly to the polarized Hilbert space. In particular, in each of the position, momentum, and Segal-Bargmann representations, the quantization of observables that are linear in the position and momentum variables generate an irreducible representation of the Heisenberg group with the same central character, and so by the Stone-von Neumann theorem~\cite{Stone1930,vonNeumann1931,Stone1931,vonNeumann1932}, they are unitarily equivalent.

\paragraph*{Gauge dependence.}

The choice of symplectic potential in Eq.~\eqref{Equation:SymplecticPotential} leads to the asymmetric-looking position and momentum operators in Eq.~\eqref{Equation:PrequantumPositionMomentumOperators}. If instead we use the more symmetric symplectic potential $\theta=\sum_j(p_jdq_j-q_jdp_j)/2$, then $q_j$ and $p_j$ prequantize to
\begin{equation}
    \begin{cases}
        \mathcal{Q}_{\mathrm{pre}}(q_j)=i\hbar\frac{\partial}{\partial p_j}+\frac{1}{2}q_j,\\
        \mathcal{Q}_{\mathrm{pre}}(p_j)=-i\hbar\frac{\partial}{\partial q_j}+\frac{1}{2}p_j.
    \end{cases}
\end{equation}
The resulting prequantum displacement operators $D_v$ for each $v\in V$ act on wavefunctions $\Psi(u)$ as
\begin{equation}\label{Equation:RealSymmetricDisplacementOperator}
    (D_v\Psi)(u)=e^{i\omega(v,u)/(2\hbar)}\Psi(u-v),
\end{equation}
and these also generate a representation of the Heisenberg group $H(V)$. This representation is related to the previous one by a unitary gauge transformation, and up to some gauge factors, the resulting spaces of polarized wavefunctions are the same. 

In particular, in the symmetric gauge, the polarized wavefunctions of the holomorphic polarization become $\Psi(q,p)=e^{-(iq\cdot p+\alpha|p|^2)/(2\hbar)}F(z)$ where $F$ is an arbitrary holomorphic function on $\mathbb{C}^n$. Equivalently, these wavefunctions can be written as
\begin{equation}\label{Equation:HolomorphicWavefunctionsSymmetricGauge}
    \Psi(z)=e^{-(z\cdot\bar{z}-z\cdot z)/(4\alpha\hbar)}F(z),
\end{equation}
and since the factor $e^{z\cdot z/(4\alpha\hbar)}$ is a non-zero holomorphic function of $z$, it can be absorbed into $F$, giving the usual Gaussian-weighted wavefunctions of the Segal-Bargmann representation $\Psi(z)=e^{-|z|^2/(4\alpha\hbar)}F(z)$.

\section*{Quantization over finite fields}

In this section, we introduce a partial analogue of geometric quantization for symplectic vector spaces over finite fields. In this setting, there is no direct analogue of differential operators, Hamiltonian vector fields, covariant derivatives, etc., and we do not construct a full quantization map. Rather, we construct representations of the discrete Heisenberg group, which have a purely algebraic description, through a procedure that has a similar structure to geometric quantization of symplectic vector spaces over the reals.

\medskip

Throughout this section, let $(V,\omega)$ be a $2n$-dimensional symplectic vector space over $\mathbb{F}:=\mathbb{F}_d$, the finite field of odd-prime order $d$, with symplectic form
\begin{equation}\label{Equation:FiniteSymplecticForm}
    \omega((q,p),(q',p')):=p\cdot q'-q\cdot p',\;\forall (q,p),(q',p')\in V.
\end{equation}
Let $\zeta=e^{2\pi i/d}$ be a primitive $d^{th}$ root of unity, and let $\chi:\mathbb{F}\to\mathbb{C}^\times$ be the character $\chi(t)=\zeta^t$.

\paragraph*{Prequantization.}
The prequantum Hilbert space is the space of complex-valued functions on the phase space, $L^2(V)\cong\mathbb{C}^{d^{2n}}$. For each $v\in V$, we can define a prequantum displacement operator $D_v$ by its action on prequantum wavefunctions $\Psi\in L^2(V)$ as
\begin{equation}
    (D_v\Psi)(u)=\chi\left(\omega(v,u)2^{-1}\right)\Psi(u-v)
\end{equation}
(cf.~Eq.~\ref{Equation:RealSymmetricDisplacementOperator}) where $2^{-1}$ is the multiplicative inverse of $2$ in $\mathbb{F}$. This is the prequantum operator associated to the observable $\chi\circ f_v$ where $f_v(\cdot):=\omega(v,\cdot)$. 
A quick calculation shows that these displacement operators satisfy the discrete Weyl relations
\begin{equation}
    \begin{cases}
        D_uD_v=\chi\left(\omega(u,v)2^{-1}\right)D_{u+v}\\
        D_{-v}=D_v^{-1}
    \end{cases}\quad\forall u,v\in V,
\end{equation}
and the operators $\{\chi(t)D_v\;|\; t\in\mathbb{F},\,v\in V\}$ define a representation of the discrete Heisenberg group $H(V)$. This representation is faithful, but it is also reducible.

\paragraph*{Polarization.}
In the real case, polarized Hilbert spaces are obtained by restricting to functions that are covariantly constant along some directions. In the finite field setting, there is no direct equivalent of these differential constraints, the closest analogue uses finite parallel transport. To this end, we introduce parallel transport operators $T_v$ defined for each $v\in V$ by their action on prequantum wavefunctions as
\begin{equation}
    (T_v\Psi)(u)=\chi(\omega(v,u)2^{-1})\Psi(u+v).
\end{equation}
These operators are chosen so that they commute with the prequantum displacement operators and they compose as
\begin{equation}\label{Equation:TransportOperatorsComposition}
    T_uT_v=\chi(-\omega(u,v)\;2^{-1})T_{u+v}\quad\forall u,v\in V.
\end{equation}

Mirroring the real case, we take the position polarization to be $\{(0,p')\;|\;p'\in\mathbb{F}^n\}$, and the corresponding subspace of polarized wavefunctions to be those that are covariantly constant along the $p_j$ directions, by which we mean functions satisfying $(T_{(0,p')}\Psi)(q,p)=\Psi(q,p)$ for all $p'\in\mathbb{F}^n$. These constraints are consistent since the polarization is a Lagrangian subspace of $V$, and so with Eq.~\eqref{Equation:TransportOperatorsComposition} the corresponding parallel transport operators commute. The resulting space of polarized wavefunctions consists of functions of the form
\begin{equation}
    \Psi(q,p)=\chi(-q\cdot p\;2^{-1})\psi(q),
\end{equation}
where $\psi:\mathbb{F}^n\to\mathbb{C}$ is arbitrary. That is, up to the gauge factor, the functions depend only on the position variables $q$. The prequantum displacement operators act on these functions like
\begin{equation}
    (D_{(q',p')}\Psi)(q,p)=\chi(-q\cdot p\;2^{-1})\chi(p'\cdot q-q'\cdot p'\;2^{-1})\psi(q-q'),
\end{equation}
and so the induced action on the functions $\psi(q)$ is
\begin{equation}\label{Equation:InducedDisplacementOperators}
    \psi(q)\xmapsto{D_{(q',p')}}\chi(p'\cdot q-q'\cdot p'\;2^{-1})\psi(q-q').
\end{equation}
This induced action exactly reproduces the displacement operators used in Gross' construction of the discrete Wigner function and discrete Weyl quantization~\cite{Gross2006}.

Similarly, the momentum polarization is $\{(q',0)\;|\;q'\in\mathbb{F}^n\}$, with polarized wavefunctions satisfying $(T_{(q',0)}\Psi)(q,p)=\Psi(q,p)$ for all $q'\in\mathbb{F}^n$. These are functions of the form $\Psi(q,p)=\chi(q\cdot p\;2^{-1})\phi(p)$.

To get an analogue of the Segal-Bargmann representation, let $\mathbb{E}=\mathbb{F}_{d^2}$ be the finite field of order $d^2$. Here the Frobenius automorphism, defined by $x\mapsto\bar{x}:=x^d$ for all $x\in\mathbb{E}$, plays the role of complex conjugation. Fix a quadratic extension $\mathbb{E}/\mathbb{F}$, and a field basis $\{1,a\}$ of $\mathbb{E}$ over $\mathbb{F}$ such that $\bar{a}=-a$. For any fixed $\alpha\in\mathbb{F}^\times$, we identify $V=\mathbb{F}^{2n}$ with $\mathbb{E}^n$ via the bijection
\begin{equation}\label{Equation:FiniteComplexification}
    \iota_\alpha:\mathbb{F}^{2n}\to\mathbb{E}^n::(q,p)\mapsto q+a\alpha p=:z.
\end{equation}
Note that in Eq.~\eqref{Equation:SegalBargmannCoordinates}, the minus sign was necessary in order to get square-integrable wavefunctions~\cite{Hall2013}, whereas here, the corresponding minus sign is not necessary and simply amounts to a different choice of $\alpha$.

Now let $\mu=q+\alpha p$ and $\nu=q-\alpha p$. Inverting these expressions we have $q=(\mu+\nu)2^{-1}$ and $p=(\mu-\nu)(2\alpha)^{-1}$, or equivalently, $(q,p)=(\mu,\alpha^{-1}\mu)2^{-1}+(\nu,-\alpha^{-1}\nu)2^{-1}$. This gives two Lagrangian subspaces of $V$,
\begin{equation}
    \begin{cases}
        L_z:=\{(w,\alpha^{-1}w)\;|\;w\in\mathbb{F}^n\}\\
        L_{\bar z}:=\{(w,-\alpha^{-1}w)\;|\;w\in\mathbb{F}^n\},
    \end{cases}
\end{equation}
that satisfy $V=L_z\oplus L_{\bar z}$. Under $\iota_\alpha$, these map to
\begin{equation}
    \begin{cases}
        U=\iota_\alpha(L_z)=(1+a)\mathbb{F}^n\\
        \bar{U}=\iota_\alpha(L_{\bar{z}})=(1-a)\mathbb{F}^n,
    \end{cases}
\end{equation}
with $\mathbb{E}^n=U\oplus_{\mathbb{F}}\bar{U}$.

Under $\iota_\alpha$, the coordinate translation $(q,p)\mapsto(q+w,p-\alpha^{-1}w)$ becomes $z\mapsto z+(1-a)w$, i.e., addition of an element of $\bar{U}$. Therefore, with polarization $L_{\bar{z}}$, the subspace of polarized wavefunctions is obtained by imposing covariant constancy in the $L_{\bar{z}}$ directions,
\begin{equation}
    (T_{(w,-\alpha^{-1}w)}\Psi)(q,p)=\Psi(q,p),
    \qquad \forall w\in\mathbb{F}^n.
\end{equation}
The general solution is
\begin{equation}\label{Equation:DiscreteHolomorphicWavefunctionMunu}
    \Psi(q,p)=\chi\left(\frac{\mu\cdot\nu}{4\alpha}\right)F(\mu),
\end{equation}
where $F:\mathbb{F}^n\to\mathbb{C}$ is arbitrary. Thus, up to the gauge factor $\chi(\mu\cdot\nu/4\alpha)$, polarized wavefunctions depend only on the $L_z$ coordinates.

With the identification of $V$ with $\mathbb{E}^n$ via $\iota_\alpha$, polarized wavefunctions have the form
\begin{equation}
    \Psi(z)=\chi\left(\frac{1}{16\alpha a^2}\Tr((a^2-1)z\cdot z+(a^2+1)\bar{z}\cdot z)\right)F([z])
\end{equation}
where $F:\mathbb{E}^n/\bar{U}\to\mathbb{C}$ is arbitrary, and the field trace $\Tr:\mathbb{E}\to\mathbb{F}$ is $\Tr(x):=x+\bar{x}$.

As in the real case, we can absorb a factor into $F$ as long as it is nonzero everywhere and constant in the $\bar{U}$ directions. Subtracting $\mu\cdot\mu/(4\alpha)$ from the argument of the prefactor in Eq.~\eqref{Equation:DiscreteHolomorphicWavefunctionMunu}, the polarized wavefunctions can be written as
\begin{equation}\label{Equation:DiscreteHolomorphicWavefunction}
    \Psi(z) = \chi\left(\frac{1}{8\alpha a^2}\Tr\left[(1+a)(\bar{z}\cdot z-z\cdot z)\right]\right)F([z]).
\end{equation}
If $d\equiv3\pmod{4}$, then we can choose $a$ such that $a^2=-1$, and then Eq.~\eqref{Equation:DiscreteHolomorphicWavefunction} simplifies to
\begin{equation}
    \Psi(z)=\chi\left(-\frac{1}{8\alpha}\Tr[(1+a)(\bar{z}\cdot z-z\cdot z)]\right)F([z])
\end{equation}
(cf.~Eq.~\eqref{Equation:HolomorphicWavefunctionsSymmetricGauge}).

\paragraph*{Quantization.}
For a polarization $L\subset V$ (i.e., a Lagrangian subspace), the space of polarized wavefunctions is
\begin{equation}
    \mathcal{H}_L=\{\Psi\in L^2(V)\;|\;T_v\Psi=\Psi\;\forall v\in L\}.
\end{equation}
For any such function,
\begin{equation}
    \Psi(u+v)=\chi\left(-\omega(v,u)2^{-1}\right)\Psi(u)\quad\forall v\in L,
\end{equation}
and hence, one value of $\Psi$ on each $L$-coset determines the full polarized wavefunction. Choosing coset representatives, we can identify $\mathcal{H}_L$ with the space of complex-valued functions on $V/L$, and so $\dim\mathcal{H}_L=|V/L|=d^n$.

To get a quantum Hilbert space, we need to equip the space of polarized wavefunctions with a suitable Hermitian inner product. In the discrete setting we do not need half-form corrections to get square-integrable wavefunctions, for any of the polarizations above, the space of polarized wavefunctions is already a non-trivial subspace of $L^2(V)$. We equip this space with the inner product
\begin{equation}
    \langle\Psi,\Phi\rangle_L=\frac{1}{d^n}\sum\limits_{v\in V}\overline{\Psi(v)}\Phi(v).
\end{equation}
Since $\overline{\Psi(v)}\Phi(v)$ is constant on each $L$-coset of $V$, this is equal to the usual inner product on $L^2(V/L)\cong\mathbb{C}^{d^n}$.

The displacement operators $D_v$ commute with all of the operators $T_v$, so they preserve the space of polarized wavefunctions. Moreover, their restrictions to this space, as in Eq.~\eqref{Equation:InducedDisplacementOperators}, give irreducible representations of the Heisenberg group $H(V)$, and so by the discrete Stone-von Neumann theorem~\cite{Neuhauser2002,Gross2006}, representations obtained from different polarizations are unitarily equivalent. Therefore, for any choice of polarization, quantization produces a $d^n$-dimensional Hilbert space with an irreducible representation of the Heisenberg group, i.e., standard finite-dimensional quantum theory.

\section*{Quantum codes are quantized codes}\label{Section:CodingTheory}

In this section, we briefly review classical and quantum codes, and we relate the two using the quantization procedure from the previous section. See, for example, Ref.~\cite{MacWilliamsSloane1977} for more background on classical coding theory, and Refs.~\cite{Gottesman1997,Gottesman2026} for more on quantum coding theory.

\medskip

\paragraph*{Classical codes.}

An $[N,k]_q$ linear code is a $k$-dimensional subspace of $\mathbb{F}_q^N$, where $q$ is a prime power. With respect to the Euclidean inner product on $\mathbb{F}_q^N$,
\begin{equation}
    \langle u,v\rangle_E=\sum\nolimits_{j}u_jv_j,
\end{equation}
the Euclidean dual code of a code $C$ (often simply called the dual code) is the $[N,N-k]_q$ code
\begin{equation}
    C^{\perp_E}:=\left\{u\in\mathbb{F}_q^N\;\Big|\;\langle u,v\rangle_E=0\;\forall v\in C\right\}.
\end{equation}
A code is called Euclidean self-orthogonal if $C\subseteq C^{\perp_E}$, and it is Euclidean self-dual if $C=C^{\perp_E}$. When $q=r^2$, we can also define a Hermitian inner product on $\mathbb{F}_q^N$,
\begin{equation}\label{Equation:HermitianProduct}
    \langle u,v\rangle_H=\sum\nolimits_{j}\bar{u}_jv_j,
\end{equation}
where $\bar{x}:=x^r$ is the Frobenius automorphism. The Hermitian dual of a code, and Hermitian self-orthogonality and Hermitian self-duality are defined similarly to the Euclidean case, and in fact, dual codes, self-orthogonality, and self-duality can be defined similarly with respect to any inner product on $\mathbb{F}_q^N$.

Parity checks are elements of the dual code, and so for any parity check $u\in C^\perp$ and any code word $v\in C$, we have $\langle u,v\rangle=0$. A set of parity checks $S\subset C^\perp$ is complete if it spans $C^\perp$, and so the code is uniquely specified by the corresponding orthogonality constraints
\begin{equation}
    C=\left\{v\in\mathbb{F}_q^N\;\Big|\;\langle u,v\rangle=0\;\forall u\in S\right\}.
\end{equation}

Additive codes are slightly more general than linear codes. An additive code over $\mathbb{F}_q$ is a subset of $\mathbb{F}_q^N$ that is closed under addition, but not necessarily $\mathbb{F}_q$-linear.

\paragraph*{Quantum codes.}

Let $d$ be an odd prime. On $n$ qudits each with Hilbert space dimension $d$, the overall Hilbert space is $(\mathbb{C}^d)^{\otimes n}\cong\mathbb{C}^{d^n}$. The Pauli group $\mathcal{P}$ is the group generated by the generalized Pauli $Z$ and $X$ operators acting on each of the $n$ qudits, where the single-qudit Pauli operators are
\begin{equation}
    Z=\sum\limits_{j\in\mathbb{F}_d}\zeta^j\ket{j}\bra{j}\quad\text{and}\quad X=\sum\limits_{j\in\mathbb{F}_d}\ket{j+1}\bra{j}.
\end{equation}
A quantum $[\![n,k]\!]_d$ stabilizer code is identified with a stabilizer group, that is, an Abelian subgroup of the Pauli group not containing $\zeta\mathbbm{1}$~\cite{Gottesman1997}. The code space is the subspace of the Hilbert space fixed by the operators in the stabilizer group. For an $[\![n,k]\!]_d$ code, the stabilizer group has order $d^{n-k}$ and the code space has dimension $d^k$.

For a code with stabilizer group $\mathcal{S}$, logical Pauli operators are elements of the normalizer of $\mathcal{S}$ in $\mathcal{P}$,
\begin{equation}
    \mathcal{N}(\mathcal{S})=\{P\in\mathcal{P}\;|\;PS=SP\;\forall S\in\mathcal{S}\}.
\end{equation}
Since, $\mathcal{S}$ is Abelian, $\mathcal{S}\subset\mathcal{N}(\mathcal{S})$, and logical Pauli operators are equivalent up to multiplication by elements of $\mathcal{S}$. Therefore, up to this equivalence, logical operators can be identified with elements of $\mathcal{N}(\mathcal{S})/\mathcal{S}$. This quotient is isomorphic to the $k$-qudit Pauli group.

\paragraph*{Quantum codes are quantized codes.}

For the remainder of this section, as in the previous section, let $\mathbb{F}=\mathbb{F}_d$ and $\mathbb{E}=\mathbb{F}_{d^2}$, where $d$ is an odd prime, and let $(V,\omega)$ be a $2n$-dimensional symplectic vector space over $\mathbb{F}$. Let $\mathcal{H}$ be a quantization of $V$ (any polarization), with irreducible Heisenberg group representation $\rho(v,t)=\chi(t)D_v$.

In the position polarization, the induced displacement operators in Eq.~\eqref{Equation:InducedDisplacementOperators} (which we also denote by $D_v$) are
\begin{equation}\label{Equation:PauliOperators}
    D_{(q',p')}=\chi\left(-\frac{q'\cdot p'}{2}\right)\bigotimes\limits_{k=1}^nZ^{p'_k}X^{q'_k},
\end{equation}
and so, as is well-known, the $n$-qudit Pauli group on odd-prime-dimensional qudits is isomorphic to the Heisenberg group $H(V)$, and its elements modulo phase are labeled by $V$~\cite{Gross2006}.

The symplectic form on $V$ computes the commutation relations of the Pauli operators via
\begin{equation}
    D_uD_v=\chi(\omega(u,v))D_vD_u,
\end{equation}
and so a stabilizer group can be specified by a pair $(I,r)$ where $I\subset\mathbb{F}^{2n}$ is an isotropic subspace and $r:I\to\mathbb{F}$ is a linear function. An $n-k$-dimensional isotropic subspace $I$ corresponds to an $[\![n,k]\!]_d$ stabilizer code with stabilizer group is $\mathcal{S}=\{\chi(-r(v))D_v\;|\; v\in I\}$.

We can define a bracket on the space $V^*\oplus\mathbb{F}$ of affine-linear observables by setting $\{f_u,f_v\}=-\omega(u,v)$ for linear observables $f_v=\omega(v,\cdot)$, and $\{c,f\}=0$ for every $c\in\mathbb{F}$ and $f\in V^*\oplus\mathbb{F}$. This mirrors the Poisson bracket for affine-linear observables in the real case. Then, for an $n-k$-dimensional isotropic subspace $I\subset V$, equations of the form $f_v=0$ for $v\in I$ can be thought of as first-class constraints. Their constraint surface is
\begin{equation}
    \{u\in V\;|\;f_v(u)=0\;\forall v\in I\}=I^\perp,
\end{equation}
i.e., the symplectic complement of $I$. Equivalently, $I^\perp$ is a classical $[2n,n+k]_d$ linear code with symplectic parity checks $\{u\;|\;u\in I\}$.

The corresponding gauge transformations are translations by elements of $I$, and so the symplectic reduction~\cite{MarsdenWeinstein1974} of the phase space is
\begin{equation}
    \tilde{V}:=I^\perp/I\cong\mathbb{F}^{2k}.
\end{equation}
The symplectic form $\omega$ restricted to $I^\perp$ is invariant under translations by elements of $I$ and has radical $I$, so it induces a non-degenerate symplectic form on $I^\perp/I$, given by $\tilde{\omega}([u],[v])=\omega(u,v)$ for all $u,v\in I^\perp$.

The constraint $f_u(v)=0$ is equivalent to $\chi(f_u(v))=1$, and under quantization, the mapping $\chi\circ f_u\longmapsto D_u$ gives the quantum constraints $D_u\Psi=\Psi$. The collection of these quantum constraints for all $v\in I$ defines a stabilizer group, $\mathcal{S}_I=\{D_u\;|\;u\in I\}$ with code space $\mathcal{C}_I=\{\Psi\in\mathcal{H}\;|\;D_u\Psi=\Psi\;\forall u\in I\}$.

For any $n-k$-dimensional isotropic subspace $I$, let $\tilde{\mathcal{H}}$ be any quantization of the reduced phase space $I^\perp/I$ with central character $\chi$ and Heisenberg group representation $\tilde{\rho}([v],t)=\chi(t)\tilde{D}_{[v]}$. The following result shows that symplectic reduction is compatible with quantization.

\begin{Theorem}\label{Theorem:MainResult}
    A displacement operator $D_v$ preserves the code space $\mathcal{C}_I$ if and only if $v\in I^\perp$, and there exists a unitary isomorphism $U_I:\tilde{\mathcal{H}}\to\mathcal{C}_I$ such that
    \begin{equation}\label{Equation:MainResultIntertwiner}
        U_I\tilde{D}_{[v]}U_I^\dagger=D_v|_{\mathcal{C}_I}\quad\forall v\in I^\perp.
    \end{equation}
\end{Theorem}

\noindent The proof is in the appendix.

\medskip

Therefore, the code space $\mathcal{C}_I$ can equivalently be described as both the subspace of the $n$-qudit Hilbert space picked out by the stabilizer constraints of an $[\![n,k]\!]_d$ stabilizer code, and as the quantization of the reduced classical phase space under a set of linear first-class constraints. Under quantization, the characters of first-class constraints map directly to stabilizers, and characters of gauge-invariant linear observables (labeled by $I^\perp$) give logical displacement operators. For linear observables labeled by elements of $I^\perp$, gauge transformations, given by translations of the labels by elements of $I$, give equivalent observables on the constraint surface and they map to the same operator on the code space.

To summarize, imposing the stabilizer constraints labeled by $I$ produces the same quantum system as quantizing the symplectic reduction $I^\perp/I$. Furthermore, the space $I^\perp$ can be identified with an $[2n,n+k]_d$ classical code with symplectic parity checks $I$. In this sense, quantum codes are quantized codes.

\medskip

Other syndrome sectors are given by nonzero constraints. To see this, start by defining the ``momentum map'' $\mu_I:V\to I^*$ by $\mu_I(u)(v)=f_v(u)$, and for any $r\in I^*$, set $\Sigma_r:=\mu_I^{-1}(r)$. We denote code spaces corresponding to nonzero syndromes $r\in I^*$ by
\begin{equation}
    \mathcal{C}_I^r=\{\Psi\in\mathcal{H}\;|\;\chi(-r(v))D_v\Psi=\Psi\;\forall v\in I\}.
\end{equation}
Then the classical phase space and the quantum Hilbert space decompose as
\begin{equation}
    V=\bigsqcup\limits_{r\in I^*}\Sigma_{r}\quad\text{and}\quad \mathcal{H}=\bigoplus\limits_{r\in I^*}\mathcal{C}_I^r.
\end{equation}
Choose a $w_r\in V$ with $\mu_I(w_r)=r$. Then $\Sigma_r=w_r+I^\perp$ and $D_{w_r}\mathcal{C}_I=\mathcal{C}_I^r$. That is, translation by $w_r$ identifies $\tilde{V}$ with the affine reduced space $\Sigma_r/I$, and $D_{w_r}$ identifies their quantum code spaces. In particular, every syndrome sector has dimension $d^k$, and the corresponding stabilizer group is $\mathcal{S}_I^r=\{\chi(-r(u))D_u\;|\;u\in I\}$. Every stabilizer group can be identified with such a pair $(I,r)$.

\medskip

There is a correspondence between quantum stabilizer codes on $d$-dimensional qudits and certain classical codes over $\mathbb{E}$~\cite{CalderbankSloane1998,Rains1999,KetkarSarvepalli2006}. Its description uses the bijection $\iota_\alpha$ defined in the previous section that identifies $\mathbb{F}^{2n}$ and $\mathbb{E}^n$.

Let $\lambda:=(2\alpha\bar{a})^{-1}\in\mathbb{E}$. Then there is an $\mathbb{F}$-bilinear alternating form on $\mathbb{E}^n$,
\begin{equation}\label{Equation:ModifiedTraceHermitian}
    \langle x,y\rangle_{TA}:=\Tr(\lambda\langle x,y\rangle_H)\in\mathbb{F}.
\end{equation}
This is called the trace-alternating (TA) form on $\mathbb{E}^n$. For an additive code $C\subset\mathbb{E}^n$, its dual code with respect to the inner product Eq.~\eqref{Equation:ModifiedTraceHermitian} is
\begin{equation}\label{Equation:MTHDualCode}
    C^{\perp_{TA}} := \left\{z\in\mathbb E^n\;\big|\;\langle c,z\rangle_{TA}=0\;\forall c\in C\right\}.
\end{equation}
Since $d$ is prime, additive subgroups of $\mathbb{E}^n$ are actually $\mathbb{F}$-linear subspaces.

Through the map $\iota_\alpha$, the standard symplectic form on $\mathbb{F}^{2n}$ is exactly the trace-alternating form on $\mathbb{E}^n$. More precisely, for all $u,v\in\mathbb{F}^{2n}$,
\begin{equation}\label{Equation:SymplecticMTHCorrespondence}
    \langle\iota_\alpha(u),\iota_\alpha(v)\rangle_{TA}=\omega(u,v).
\end{equation}
This follows from a short calculation. For $(q,p),(q',p')\in\mathbb{F}^{2n}$, we have
\begin{align}
    \langle\iota_\alpha(q,p),&\iota_\alpha(q',p')\rangle_H = (q-a\alpha p)\cdot(q'+a\alpha p')\\
    &= q\cdot q'-a^2\alpha^2p\cdot p'+a\alpha(q\cdot p'-p\cdot q').
\end{align}
Since $\lambda+\bar{\lambda}=0$, after multiplying the above expression by $\lambda$ and taking a trace, the first two terms vanish. Then, since $\lambda a\alpha=-2^{-1}$, we are left with
\begin{align}
    &\Tr(\lambda\langle\iota_\alpha(q,p),\iota_\alpha(q',p')\rangle_H)&\\
    &= \Tr\left[-2^{-1}(q\cdot p'-p\cdot q')\right]
    = p\cdot q'-q\cdot p'.\label{Equation:SymplecticMTHCorrespondenceEnd}
\end{align}

Therefore, through $\iota_\alpha$, an additive trace-alternating self-orthogonal code of order $d^{n-k}$ in $\mathbb{E}^n$ corresponds to an $n-k$-dimensional isotropic subspace of $V$, and hence, identifies an $[\![n,k]\!]_d$ stabilizer code. This gives the following established result.

\begin{Theorem*}[\cite{CalderbankSloane1998,Rains1999,KetkarSarvepalli2006}]\label{Theorem:SegalBargmannCodes}
    There is a correspondence between $n-k$-dimensional isotropic subspaces of $V$ and TA-self-orthogonal additive codes in $\mathbb{E}^n$ of order $d^{n-k}$.
\end{Theorem*}

\noindent With Theorem~\ref{Theorem:MainResult}, this gives the following corollary.

\begin{Corollary}\label{Corollary:SegalBargmannCodes}
    Let $I\subset V$ be an $n-k$-dimensional isotropic subspace with corresponding additive trace-alternating self-orthogonal code $C=\iota_\alpha(I)\subset\mathbb{E}^n$. The map $\iota_\alpha$ induces a symplectic isomorphism
    \begin{equation}
        I^\perp/I\to C^{\perp_{TA}}/C,\quad [v]\mapsto[\iota_\alpha(v)].
    \end{equation}
    Therefore, the code space $\mathcal{C}_I$, together with its logical Heisenberg group is a quantization of $C^{\perp_{TA}}/C$.
\end{Corollary}

\noindent The proof is in the appendix.

\pagebreak

The map $\iota_\alpha$ which identifies $\mathbb{F}^{2n}$ with $\mathbb{E}^n$ that we used in the previous section to define the discrete analogue of the Segal-Bargmann representation is the same structure used to define the correspondence between classical trace-alternating self-orthogonal codes over $\mathbb{E}$ and quantum codes on $d$-dimensional qudits~\cite{CalderbankSloane1998,Rains1999,KetkarSarvepalli2006}. In other words, trace-alternating self-orthogonal additive codes over $\mathbb{E}$ are isotropic subspaces of $V$ written in the Segal-Bargmann coordinates of Eq.~\eqref{Equation:FiniteComplexification}, and under quantization, they also yield stabilizer codes.

The conclusion of Theorem~\ref{Theorem:MainResult} and Corollary~\ref{Corollary:SegalBargmannCodes}, is that stabilizers correspond to the finite-field analogue of a set of first-class constraints, with $I^\perp/I$ the reduced classical phase space, and the logical Pauli group is the Heisenberg group associated to this reduced phase space. The method of constructing quantum codes from trace-alternating self-orthogonal classical codes is obtained by the same process after writing the phase space in the discrete analogue of Segal-Bargmann coordinates Eq.~\eqref{Equation:FiniteComplexification}.

\section*{Discussion}\label{Section:Discussion}

In this paper, we introduced a construction that relates quantum stabilizer codes to quantization of constrained linear systems over finite fields. A certain classical code, given by a set of linear parity checks $I$, identifies a set of classical first-class constraints with constraint surface $I^\perp$ and gauge directions $I$. Under quantization, the classical constraints map to stabilizer constraints, and quantizing the symplectic reduction of the phase space gives the code space and logical operators of the stabilizer code.

In a bit more detail, we introduce a method for constructing representations of discrete Heisenberg groups over finite fields $\mathbb{F}_d$ of odd-prime order $d$. This method closely mirrors the quantization of observables that are linear in the position and momentum variables in geometric quantization of finite-dimensional symplectic vector spaces over the reals.

We then show that this discrete quantization procedure is compatible with imposing certain constraints, in the sense that imposing these constraints and then quantizing is equivalent to quantizing first and then imposing the corresponding quantum constraints. This is the idea behind the slogan \emph{quantization commutes with reduction} from geometric quantization~\cite{GuilleminSternberg1982,TianZhang1998}.

On the classical side, the relevant constraints are discrete analogues of linear first-class constraints (in Dirac's classification~\cite{Dirac1964}), and we point out that these can be identified with the parity checks of a classical linear code over $\mathbb{F}_d$. This identifies the constraint surface with the code itself, and the symplectic reduction is obtained after identifying points that are related by a gauge transformation. Under quantization, the classical constraints (i.e.~parity checks) map to stabilizer constraints of a quantum code, and characters of linear observables on the constraint surface map to logical Pauli operators, hence the conclusion: \emph{quantum codes are quantized codes.}

Finally, using coordinates coming from a polarization that mirrors that of the Segal-Bargmann representation~\cite{Bargmann1961,Segal1963} in the real case, we reproduce a known correspondence between quantum stabilizer codes on $d$-dimensional qudits and classical trace-alternating self-orthogonal codes over $\mathbb{F}_{d^2}$~\cite{Rains1999,KetkarSarvepalli2006}. This provides a more physical interpretation for a previously established quantum-classical correspondence in coding theory.

There are some limitations to these results which leave room for future work. First, we only discuss the case of odd-prime qudit dimension $d$. Some results may extend to odd-prime-power dimensions, but there is some subtlety about whether the resulting Pauli group up to phases is defined to be isomorphic to $\mathbb{F}_d^{2n}$ or $\mathbb{Z}_d^{2n}$~\cite{Gross2006} (for quantum codes, this amounts to the difference between so-called Galois qudits and modular qudits~\cite{AlbertFaist2026}). Some results may also go through for even dimensions, for example, the correspondence between classical and quantum codes also holds for classical codes over $\mathbb{F}_4$ and stabilizer codes on qubits~\cite{CalderbankSloane1998}, though complications arise in representation theory and coding theory in even dimensions and with fields of characteristic $2$~\cite{GurevichHadani2012,RaussendorfFeldmann2023}.

Second, unlike in the real case, in the discrete setting we only discuss observables of a particular type, namely, characters of functions that are linear in the position and momentum variables. This is all that is required to define quantum codes, so it suffices for our purposes. However, these observables form a basis for all complex-valued functions on the finite phase space, and so a complete quantization procedure could be obtained by extending the current procedure by linearity, resulting in a discrete Weyl quantization procedure~\cite{Gross2006}. We leave a more thorough analysis of this procedure for future work.

Lastly, it should be noted that the main purpose of this paper is to provide a physical and intuitive explanation for previously established technical results. For symplectic vector spaces over finite fields of odd-prime characteristic, a canonical quantization functor has been constructed, and it has been shown to be compatible with symplectic reduction~\cite{GurevichHadani2009} (canonical is used here in the mathematical sense, this is not related to what is called ``canonical quantization'' in physics). The correspondence between stabilizer codes on $d$-dimensional qudits and certain additive self-orthogonal codes over $\mathbb{F}_{d^2}$ is also established and well-known~\cite{CalderbankSloane1998,Rains1999,KetkarSarvepalli2006}. The main point of this paper, and my original motivation for writing it, is to provide a more physical interpretation for those mathematical results. To this end, we relate the former to ideas in coding theory, and the latter to quantization of constrained systems. This type of interpretational work is likely to become increasingly important as modern tools make technical mathematical results more quickly and readily available.

\section*{Acknowledgements}

I am grateful to Caroline Barton, Jack Davis, Alex Kazachek, and Sam Winnick for discussions. This work was funded by the Natural Sciences and Engineering Research Council of Canada~(NSERC).

\bibliography{biblio}

\onecolumngrid
\appendix
\section*{Proofs}

\begin{proof}[Proof of Theorem~\ref{Theorem:MainResult}]
    The projector onto the code space $\mathcal{C}_I$ is given by
    \begin{equation}
        \Pi_I=\frac{1}{|I|}\sum\limits_{u\in I}D_u,
    \end{equation}
    with $\Tr(\Pi_I)=d^n/|I|=d^k$. For any $v\in I^\perp$, $D_v$ commutes with every stabilizer, and so it commutes with $\Pi_I$ and preserves the code space $\mathcal{C}_I$. On the other hand, if $v\notin I^\perp$, then there exists a $u\in I$ such that $\omega(u,v)\ne0$, and so $D_v$ maps the code space to a different eigenspace of $D_u$ (i.e., a different syndrome sector), and in particular, $D_v$ does not preserve the code space.

    For any $v\in I^\perp$ and $w\in I$, we have
    \begin{equation}
        D_{v}\big|_{\mathcal{C}_I}=D_{v+w}\big|_{\mathcal{C}_I},
    \end{equation}
    and so $\tilde{R}_{[v]}:=D_v\big|_{\mathcal{C}_I}$ is well-defined for all $[v]\in I^\perp/I$. Furthermore, these operators satisfy
    \begin{equation}
        \tilde{R}_{[u]}\tilde{R}_{[v]}=\chi\left(\tilde{\omega}([u],[v])2^{-1}\right)\tilde{R}_{[u+v]}
    \end{equation}
    for all $[u],[v]\in I^\perp/I$. This is an representation of the Heisenberg group $H(\tilde{V})$ with central character $\chi$. All irreducible representations of $H(\tilde{V})$ with central character $\chi$ have dimension $d^k$, and so by complete reducibility, this representation is irreducible. Therefore, the existence of the intertwiner in Eq.~\eqref{Equation:MainResultIntertwiner} follows from the discrete Stone-von Neumann theorem~\cite{Neuhauser2002}.
\end{proof}

\begin{proof}[Proof of Corollary~\ref{Corollary:SegalBargmannCodes}]
    The correspondence between isotropic subspaces of $V$ and additive TA-self-orthogonal codes in $\mathbb{E}^n$ follows immediately from the Eqs.~\eqref{Equation:SymplecticMTHCorrespondence}--\eqref{Equation:SymplecticMTHCorrespondenceEnd}. With Eq.~\eqref{Equation:MTHDualCode} we also get $\iota_\alpha(I^\perp)=C^{\perp_{TA}}$. Since $\iota_\alpha$ is an $\mathbb F$-linear bijection, it induces the stated symplectic isomorphism of the quotients. Then Corollary~\ref{Corollary:SegalBargmannCodes} follows from Theorem~\ref{Theorem:MainResult}.
\end{proof}

\end{document}